\documentclass[12pt,final,onecolumn]{IEEEtran}
\usepackage{lineno,hyperref}

\usepackage{amsmath,amsthm}
\usepackage{amsthm,bm}
\usepackage{amssymb}
\usepackage{amsfonts}
\usepackage{dsfont}
\usepackage{graphicx}
\usepackage{enumerate}
\usepackage[dvipsnames]{xcolor}
\usepackage{epstopdf}
\usepackage{bm}
\usepackage{float}
\usepackage{graphicx}
\theoremstyle{plain}
\newtheorem{theorem}{Theorem}

\newtheorem{proposition}{Proposition}
\newtheorem{lemma}{Lemma}

\theoremstyle{definition}
\newtheorem{definition}{Definition}

\newtheorem{example}{Example}

\theoremstyle{remark}

\def\Re{\mbox{$\text{\rm Re}\;$}}
\def\Im{\mbox{$\text{\rm Im}\;$}}

\def\cls{\mbox{$\,\#\,$}}

\def\mr+{\mbox{$\text{mr}_{+}$}}

\graphicspath{{./fig/}}
\begin{document}
	\IEEEoverridecommandlockouts
	
	\title{When Small Gain Meets Small Phase}
	\author{Di Zhao, Wei Chen, and Li Qiu%
		\thanks{
			Di Zhao is with the Department of Control Science \& Engineering and Shanghai Institute of Intelligent Science and Technology, Tongji University, Shanghai, China (e-mail: dzhao925@tongji.edu.cn)\\
			Wei Chen is with the Department of Mechanics and Engineering Science \& Beijing Innovation Center for Engineering Science and Advanced Technology, Peking University, Beijing, China (e-mail: w.chen@pku.edu.cn)\\
			Li Qiu is with the Department of Electronic and Computer Engineering, The Hong Kong University of Science and Technology, Clear Water Bay, Kowloon, Hong Kong S.A.R., China (e-mail: eeqiu@ust.hk)
		}
	}
	
	
	\maketitle
	\IEEEpeerreviewmaketitle
	
	\begin{abstract}
		
		In this paper, we investigate the feedback stability of multiple-input multiple-output linear time-invariant systems with combined gain and phase information. To begin with, we explore the stability condition for a class of so-called easily controllable systems, which have small phase at low frequency ranges and low gain at high frequency. Next, we extend the stability condition via frequency-wise gain and phase combination, based on which a mixed small gain and phase condition with necessity, called a small vase theorem, is then obtained. Furthermore, the fusion of gain and phase information is investigated by a geometric approach based on the Davis-Wielandt shell. Finally, for the purpose of efficient computation and controller synthesis, we present a bounded \& sectored real lemma, which gives state-space characterization of combined gain and phase properties based on a triple of linear matrix inequalities.
		
	\end{abstract}
	
	\section{Introduction}
	In the classical frequency domain analysis of single-input single-output (SISO) linear time-invariant (LTI) systems, the magnitude (gain) response
	and phase response go hand in hand. In the era of information and networks, more and more emphasis has been put on the research about multiple-input multiple-output (MIMO) systems. For such systems, the concept of magnitudes has been naturally extended and intensively studied by using the singular values of matrices, while the concept of phases is not paid enough attention to. There have been attempts to characterize phases of MIMO systems, including \cite{Owens1984NumericalR,JieChen1998gainphase,Tits1999RobustPhase,wei2021phaseLTI,chao2020nonlinear}. We refer readers to \cite{wei2021phaseLTI,chao2020nonlinear} for detailed literature related to developments of phases of systems. Recently, \cite{wei2021phaseLTI} has proposed a definition of system phase for sectorial MIMO LTI systems, based on the canonical angles of matrices \cite{Johnson1974209,WANG2020PhaseMath}.  According to  \cite{wei2021phaseLTI}, for a sectorial LTI system with $n$ inputs and $n$ outputs, $n$ frequency-dependent phases can be defined, as the counterpart to the $n$ magnitudes defined via its singular values. Moreover, a small phase theorem, as a counterpart to the well-known small gain theorem \cite{Zames1966SmallGain}, for the stability of feedback interconnected LTI systems is obtained, which states that the feedback system is stable if the sum of system phases in the loop  (or, simply loop-phase)  is less than $\pi$. 
	
	In many real-world problems, it is conservative or even unrealistic to carry out stability analysis with pure gain or phase information. For example, consider a matrix second-order system with collocated actuators and
	speed and position sensors \cite{balas1982,gardiner1992stabilizing}, which is described by the transfer function
	\begin{align}\label{eq:example}
		G(s)=(Hs+H)(Ms^2+Cs+K)^{-1}H^T,
	\end{align}
	where $M$ and $K$ are positive definite and $C$ is accretive. The phase of such a system, if existing, can be very large at the high frequency range and so be its gain at the low (or simply, its DC gain). In this case, neither of the small gain or small phase theorem is applicable to analysing the feedback stability between a pair of such systems. 
	On the other hand, it has been shown in \cite{Hara2007EasyControl} that an easily controllable system usually has small phase at low frequency range and small gain at high frequency range. Motivated by the above issues, in this paper, we propose to analyze the feedback stability of MIMO LTI systems with a mixed gain/phase approach. To start with, we first develop a mixed small gain and phase theorem for feedback stability, where the systems in a feedback loop satisfy a phase condition at low frequency ranges and a gain condition at high frequency. Such a result is then generalized by frequency-wise combination of gain and phase. Later on, we refine the result and obtain a stability condition with necessity, namely a small vase theorem. Beyond direct combination of gain and phase information in the frequency domain, we further develop stability conditions using a geometric approach and introduce concepts including constrained phases and constrained gains. Such a method is based on the Davis-Wielandt (DW) shell \cite{davis1968shell,lestas2012large}, which is a 3D generalization to the numerical range. Finally, we derive a bounded \& sectored real lemma, which gives state-space conditions for mixed gain-phase bounded systems in terms of linear matrix inequalities (LMIs). Technically, the above results are established on the basis of some well-known theories and tools, including the homotopy methods, the feedback stability conditions via integral quadratic constraints (IQCs) \cite{Megretski1997IQC}, and the generalized KYP lemma \cite{Hara2005GeneralizedKYP}. 
	
	There have been early studies on combining gain and phase information for analysis and control. The dissipativity theory \cite{Willems1972dissipativity} provides a systematic way to combine the small gain and passivity theories in a unified framework. Later on, a frequency-domain combination of small gain and passivity properties has been investigated in \cite{brian2007mix}, giving rise to several feedback stability criteria for LTI systems. The multi-variable gain-phase and sensitivity relations have been established in \cite{JieChen1998gainphase}, as an extension of the classical Bode integral relations to multi-variable systems. Our proposed study differs from the previous research mainly in that we utilize take advantages of the phase notion that has generalized the system passivity (e.g., see \cite{brian2007mix}) and is distinguished from the other definition of phases (e.g., see \cite{JieChen1998gainphase}). 
	
	The rest of the paper is organized as follows. In Section~\ref{sec:pre}, we introduce the basic notation on matrices and systems, and preliminary results on the feedback stability conditions involving gains and phases. Later in Section~\ref{sec:mixed thm}, we develop several stability conditions by directly combining frequency-wise gain and phase information. Furthermore, we take advantages of DW-Shells to establish stability conditions in Section~\ref{sec:DW shell}. Finally in Section~\ref{sec:bounded_sectored_lemma}, we develop a bounded \& sectored real lemma for the state-space characterization of mixed phase and gain properties.
	
	\section{Preliminary}\label{sec:pre}
	
	\subsection{Phases of Matrices}\label{subsec:matrix phase}
	Let $\mathbb{F} = \mathbb{R}$ or $\mathbb{C}$ be the real or complex field,  and ${\mathbb F}^n$ be the linear space of $n$-tuples of $\mathbb{F}$ over the field $\mathbb{F}$. The identity matrix in $\mathbb{C}^{n\times n}$ is denoted as $I_n$. The Euclidean norm of a vector $x\in{\mathbb F}^n$ is denoted by $\|x\|$. The phase of a complex number $c$ is denoted by $\angle c\in(-\pi,\pi]$, the real part is by $\Re\,c$, the imaginary part is by $\Im\,c $, and the conjugate is by $\bar{c}$. The singular values of a matrix $A\in{\mathbb C}^{n\times n}$ are denoted by $\sigma_k(A)$, $k=1,2,\dots  ,n$, satisfying
	$$\bar{\sigma}(A)=\sigma_1(A)\geq\sigma_2(A)\geq\cdots\geq\sigma_n(A) = \underline{\sigma}(A),$$
	and the spectral norm of $A$ is by $\|A\|:=\bar{\sigma}(A)$. The transpose and conjugate transpose of $X\in\mathbb{C}^{m\times n}$ is denoted as $X^T$ and $X^*$, respectively. The Kronecker product of $X\in\mathbb{C}^{m\times n}$ and $Y\in\mathbb{C}^{p\times q}$ is denoted as $X\otimes Y$. 
	
	In what follows, we introduce the concept of matrix phases, as the counterpart to the singular values. A detailed development towards the matrix phase is referred to \cite{WANG2020PhaseMath}. 
	A matrix $A\in{\mathbb C}^{n\times n}$ is said to be \textbf{sectorial} if its numerical range
	$${\mathcal W}(A):=\{x^*Ax:~x\in{\mathbb C}^n,~\|x\|=1\}$$ does not contain $0$. It is shown in \cite{Horn1959Unitary} that a sectorial matrix is congruent, unique up to a permutation, a diagonal unitary matrix, i.e., there exists a nonsingular matrix $T$ and a diagonal unitary matrix $D$ such that $A=T^*DT$. Such a factorization is called a sectorial decomposition. We define the phases of $A$, denoted by
	$$\bar{\phi}(A):=\phi_1(A)\geq \phi_2(A)\geq \cdots \geq \phi_n(A)=:\underline{\phi}(A),$$
	to be the phases of the eigenvalues (i.e., diagonal elements) of $D$ so that $\bar{\phi}(A)-\underline{\phi}(A)<\pi$ and $$\phi_c(A):=\dfrac{\bar{\phi}(A)+\underline{\phi}(A)}{2},$$ called the phase center of $A$, lies in $(-\pi,\pi]$. Denote by $\Psi(A):=[\underline{\phi}(A),\bar{\phi}(A)]$ the phase sector of $A$ that characterizes its phase spread on the complex plane. A sectorial matrix $A$ is called accretive if $\Psi(A)\subset(-\pi/2,\pi/2)$. 
	
	A matrix is said to be semi-sectorial if the interior of its numerical range ${\mathcal W}(A)$ does not contain 0. A semi-sectorial matrix admits the following generalized sectorial decomposition \cite{Furtado2001LAA}:
	$$A = T^*\begin{bmatrix}
		0_{n-r} & 0 & 0\\
		0 & D & 0\\
		0 & 0 & E
	\end{bmatrix} T,$$ 
	where 
	$$D = \text{diag}(e^{j\theta_1},\dots,e^{j\theta_m}),~E=\text{diag}\left(e^{j\theta_0}\begin{bmatrix}
		1 & 2\\ 0 & 1
	\end{bmatrix},\dots,e^{j\theta_0}\begin{bmatrix}
		1 & 2\\ 0 & 1
	\end{bmatrix}\right)\in\mathbb{C}^{(r-m)\times (r-m)},$$
	and $\bar{\phi}(A)= \theta_0+\pi/2 \geq \theta_1 \geq \cdots \geq \theta_m \geq \theta_0-\pi/2 = \underline{\phi}(A)$. 
	
	%

	\subsection{Phases of Systems}
	Denote by $\mathcal{RH}_\infty$ the set of all stable real-rational transfer functions. For a stable MIMO LTI system $G\in\mathcal{RH}_\infty^{n\times n}$, it is said to be (frequency-wise) sectorial if $G(j\omega)$ is sectorial for all $\omega\in[-\infty,\infty]$. For a sectorial system $G$, the vector of phases $\phi(G(j\omega))$, which is called the phase response of $G$, is element-wise continuous with $\omega\in[-\infty,\infty]$. 
	Similar to the role of the largest singular value $\bar{\sigma}(G(j\omega))$ in the analysis of feedback stability via gains, the phase sector $\Psi(G(j\omega))=[\underline{\phi}(G(j\omega)),\bar{\phi}(G(j\omega))]$ is essential in the analysis via phases, as is shown momentarily. Moreover, from the conjugate symmetric property of a real rational transfer matrix over imaginary axis, the phase information of such a system $G$ is completely contained in its half-frequency spectrum, namely, $\{{\phi}(G(j\omega)):~\omega\in[0,\infty]\}$. 
	
	{Denote a semi-circle on the complex plane with center $jq$ and radius $\epsilon\geq0$ by $SC_\epsilon(jq):=\{s\in\mathbb{C}:~|s-jq|=\epsilon,~\Re\,s >0\}$. Given $-\infty\leq\omega_{-\infty}<\omega_\infty\leq\infty$, a contour parameterized by an ordered set $j\Omega=\{j\omega_1,\dots,j\omega_m\}$ with $\omega_{-\infty}+\epsilon<\omega_1<\omega_2<\cdots<\omega_m<\omega_{\infty}-\epsilon$ and $\epsilon\geq0$ is defined as
		\begin{multline}\label{eq:contour}
			CT_\epsilon(j\Omega):=j[\omega_{-\infty},\omega_1-\epsilon]\cup SC_\epsilon(j\omega_1)
			\cup j[\omega_1+\epsilon,\omega_2-\epsilon]\cup SC_\epsilon(j\omega_2)\cup\\
			\cdots \cup j[\omega_{m-1}+\epsilon,\omega_m-\epsilon]\cup SC_\epsilon(j\omega_m)\cup j[\omega_m+\epsilon,\omega_\infty].
		\end{multline}
		In particular, the contour is denoted as $CT^\infty_\epsilon(j\Omega)$ when $\omega_{-\infty}=-\infty$ and $\omega_\infty=\infty$. }
	
	An $n\times n$ real rational proper system $G$ is said to be semi-stable if it may have poles on the imaginary axis but no poles in the open right half plane. Let $j\Omega_p$ be the set of poles and $j\omega_z$ be the set of zeros of $G$ on the imaginary axis. 
	Such a system $G$ is said to be frequency-wise semi-sectorial over $(\omega_a,\omega_b)$ if 
	\begin{itemize}
		\item [(a)]	$G(j\omega)$ is semi-sectorial for all $\omega\in(\omega_a,\omega_b)\setminus \Omega_p$; and
		\item [(b)] there exists an $\epsilon^*>0$ such that for all $\epsilon\in(0,\epsilon^*]$, $G(s)$ has a constant rank and is semi-sectorial along the contour $CT_\epsilon[(j\Omega_p\cup j\Omega_z) \cap j(\omega_a,\omega_b)]$.
	\end{itemize}
	The range $(\omega_a,\omega_b)$ will be omitted if $G$ is frequency-wise semi-sectorial over $[-\infty,\infty]$. Throughout this study, assume also that the phase center $\phi_c(G(s))$ satisfies $\phi_c(G(0))=0$ (or $\phi_c(G(\epsilon))=0$ if $0$ is a pole) and $\phi_c(G(s))=0$ is continuously defined along any contour $s\in CT_\epsilon(j\Omega)$. 
	
	\subsection{Small Gain and Small Phase Theorems}\label{subsec:small_gainPhase}
	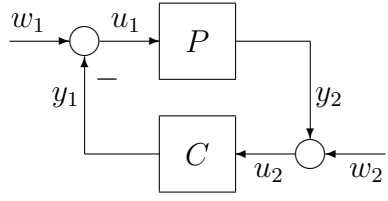
\begin{figure}[!t]
		\centering
		\setlength{\unitlength}{1mm}
		\begin{picture}(50,25)
			\put(0,20){\vector(1,0){8}} \put(10,20){\circle{4}}
			\put(12,20){\vector(1,0){8}} \put(20,15){\framebox(10,10){$P$}}
			\put(30,20){\line(1,0){10}} \put(40,20){\vector(0,-1){13}}
			\put(38,5){\vector(-1,0){8}} \put(40,5){\circle{4}}
			\put(50,5){\vector(-1,0){8}} \put(20,0){\framebox(10,10){$C$}}
			\put(20,5){\line(-1,0){10}} \put(10,5){\vector(0,1){13}}
			\put(5,10){\makebox(5,5){$y_1$}} \put(40,10){\makebox(5,5){$y_2$}}
			\put(0,20){\makebox(5,5){$w_1$}} \put(45,0){\makebox(5,5){$w_2$}}
			\put(13,20){\makebox(5,5){$u_1$}} \put(32,0){\makebox(5,5){$u_2$}}
			\put(10,10){\makebox(6,10){$-$}}
		\end{picture}
		\caption{The feedback system $P\,\#\,C$.} \label{fig:feedback}\vspace*{-10pt}
	\end{figure}
	Suppose $P$ and $C$ are $n\times n$ real rational transfer matrices. As shown in Fig.~\ref{fig:feedback}, the feedback interconnection of $P$ and $C$ is said to be stable if the Gang of Four transfer matrix
	$$P\,\#\,C:=\begin{bmatrix}I \\ P\end{bmatrix}(I+CP)^{-1}\begin{bmatrix}I & C\end{bmatrix}$$
	is stable, i.e., $P\,\#\,C\in\mathcal{RH}_\infty^{2n\times 2n}$. For notational simplicity, the feedback system is also denoted as $P\,\#\,C$. 
	
	The well-known small gain theorem \cite[Theorem~8.1]{Zhou1998Essential} can be stated in the following form. 
	\begin{lemma}\label{lem:smallGain}
		Let $P$ and $C\in\mathcal{RH}_\infty^{n\times n}$. Then the feedback system $P\,\#\,C$ is stable if
		$$\bar{\sigma}[P(j\omega)]\bar{\sigma}[C(j\omega)]<1,$$
		for all $\omega\in[-\infty,\infty]$. 
	\end{lemma}
	
	A counterpart to the above result, from the perspective of phases of systems, has been obtained in \cite{wei2021phaseLTI}. 
	\begin{lemma}\label{lem:smallPhase}
		Let $P$ be {semi-stable frequency-wise semi-sectorial with $j\Omega_p$ being the set of poles on the imaginary axis} and $C\in\mathcal{RH}_\infty^{n\times n}$ be frequency-wise sectorial. Then $P\,\#\,C$ is stable if
		$$\bar{\phi}(P(j\omega))+\bar{\phi}(C(j\omega))<\pi~~\text{and}~~\underline{\phi}(P(j\omega))+\underline{\phi}(C(j\omega))>-\pi,$$
		for all $\omega\in[0,\infty]\setminus\Omega_p$. 
	\end{lemma}
	
	
	
	\section{Feedback Stability with Mixed Gain/Phase Conditions}\label{sec:mixed thm}
	In this section, we first develop a mixed small gain/phase theorem for LTI feedback systems that satisfy a small phase condition at low frequency and a small gain condition at high frequency. Later on, we extend the result to a stability condition with necessity --- a small vase theorem. 
	
	\subsection{Mixed Gain and Phase Study on Matrices}
	Singular values are widely accepted as the gains for matrices. On the other hand, recall that we have defined phases for sectorial matrices in Section~\ref{subsec:matrix phase}. A combination of the notion of gain and phase at the level of matrices gives us the following results on invertibility of matrices. 
	
	Let $\alpha,\beta\in(-\pi,\pi)$ and $\gamma\in(0,\infty)$. Denote by
	\begin{equation}\label{eq:sets}
		\begin{aligned}
			&{\mathcal P}(\alpha,\beta):=\{A\in{\mathbb C}^{n\times n}:~A~\text{is sectorial},~\underline{\phi}(A) > \alpha,~\bar{\phi}(A) < \beta\},\\
			&\bar{{\mathcal P}}(\alpha,\beta):=\{A\in{\mathbb C}^{n\times n}:~A~\text{is sectorial},~\underline{\phi}(A)\geq \alpha,~\bar{\phi}(A) \leq \beta\},\\
			&{\mathcal G}(\gamma):=\{A\in{\mathbb C}^{n\times n}:~\bar{\sigma}(A)< \gamma\},\\
			&\bar{{\mathcal G}}(\gamma):=\{A\in{\mathbb C}^{n\times n}:~\bar{\sigma}(A)\leq \gamma\}.
		\end{aligned}
	\end{equation}
	In particular, for a scalar $c\in\mathbb{C}$, if $c\in{\mathcal P}(\alpha,\beta)\cap{\mathcal G}(\gamma)$, such a $c$ belongs to a fan-shaped region on the complex plane as is shown in Fig.~\ref{fig:fan}. Similarly, if $c\in{\mathcal P}(\alpha,\beta)\cup{\mathcal G}(\gamma)$, it belongs to a vase-shaped region as in Fig.~\ref{fig:vase}. As such, we continue to use the names fan-shaped or vase-shaped regions/sets even for complex matrices and transfer matrices when they belong to the intersection or union of the above sets in \eqref{eq:sets}. 
	\begin{figure}[H]
		\begin{minipage}{0.48\textwidth}
			\centering
			\includegraphics[width=.7\linewidth]{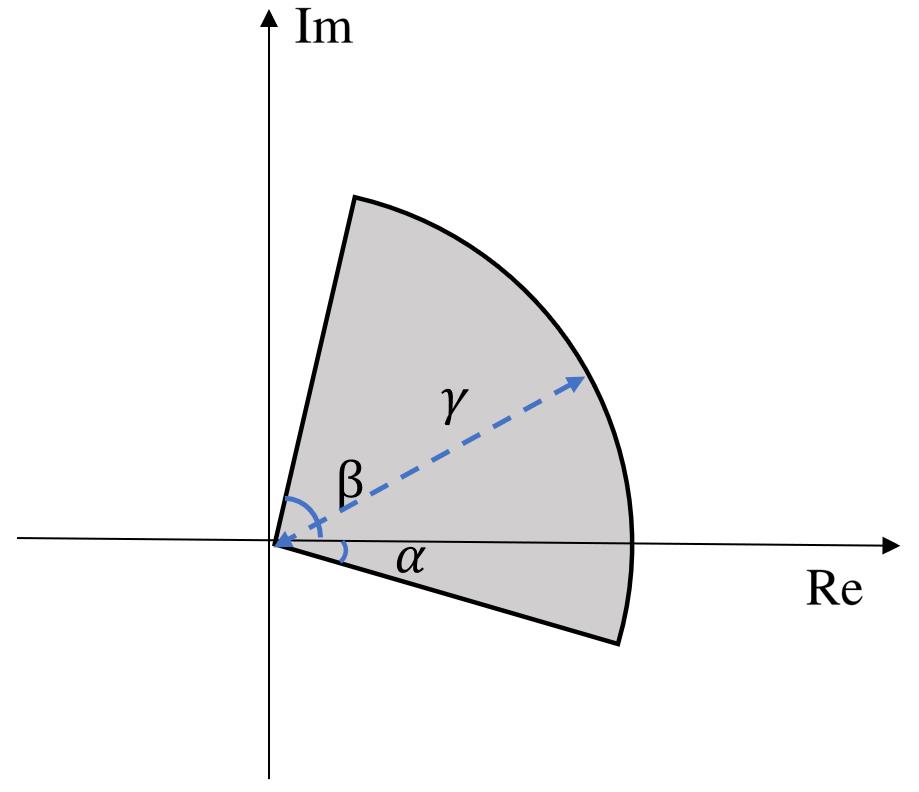}
			\caption{A fan-shaped region.}\label{fig:fan}
		\end{minipage}\hfill
		\begin{minipage}{0.48\textwidth}
			\centering
			\includegraphics[width=.7\linewidth]{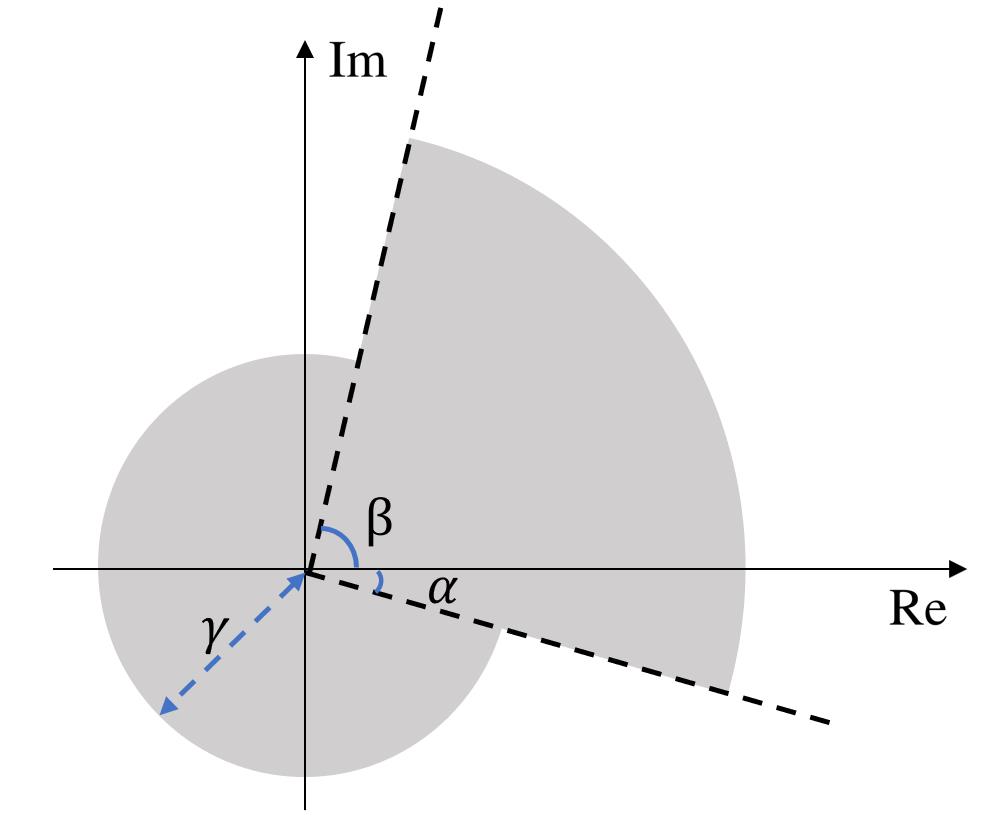}
			\caption{A vase-shaped region.}\label{fig:vase}
		\end{minipage}
	\end{figure}
	
	With the above notion, we have the following matrix invertibility results concerning the fan-shaped and vase-shaped sets, which were first investigated in \cite{WANG2020PhaseMath}. 
	\begin{lemma}\label{lem:mixGP_matrix}
		Let $A, B\in\mathbb{C}^{n\times n}$. Then $I+AB$ is invertible if there exist $\gamma\in(0,\infty)$, $\alpha,\beta\in(-\pi,\pi)$ and $\beta-\alpha\in(0,\pi]$ such that
		\begin{align}\label{eq:lem_matrix_1}
			A\in\bar{\mathcal P}(\alpha,\beta)\cup\bar{{\mathcal G}}(\gamma)~~\text{and}~~B\in{{\mathcal P}}(-\pi-\alpha,\pi-\beta)\cap{{\mathcal G}}(\gamma^{-1}).
		\end{align}
	\end{lemma}
	The above condition can be further strengthened, and then results in the following necessary and sufficient condition. 
	\begin{lemma}\label{lem:mixGP_matrix_nec}
		Let $A, B\in\mathbb{C}^{n\times n}$, $B$ be sectorial, $\gamma\in(0,\infty)$, $\alpha,\beta\in(-\pi,\pi)$ and $\beta-\alpha\in(0,\pi]$. Then $I+AB$ is invertible for all $A\in\bar{\mathcal P}(\alpha,\beta)\cup\bar{{\mathcal G}}(\gamma)$ if and only if
		\begin{align}\label{eq:lem_matrix_2}
			B\in{{\mathcal P}}(-\pi-\alpha,\pi-\beta)\cap{{\mathcal G}}(\gamma^{-1}).
		\end{align}
	\end{lemma}
	As we see in the above theorem, $A$ belongs to the union set of one set with bounded phase and the other with bounded gain. With the intuition on Fig.~\ref{fig:vase} for scalars, we come up with the name for the result --- a small vase theorem. Later on, we will further explore such results for LTI systems. 
	
	\subsection{Mixed Gain-Phase Stability with Cut-off Frequency}
	As we have introduced in Section~\ref{subsec:small_gainPhase}, the well-established small gain and small phase theorems on the stability of feedback systems, in what follows we develop stability results via the combination of gains and phases. 
	
	It is discussed in \cite{Hara2007EasyControl} that an easily controllable system usually satisfies desirable phase/gain conditions, i.e., its transfer matrix has a low gain in high frequency range but a small phase shift in low range. Many real-world systems, such as mechanical systems, satisfy this condition, which motivates us to develop the following mixed gain/phase stability result for MIMO LTI systems. 
	\begin{theorem}[Mixed Gain-Phase with Cut-off Frequency]\label{thm:mixture}
		Let $\omega_c\in(0,\infty)$, $P$ be {semi-stable frequency-wise semi-sectorial over $(-\omega_c,\omega_c)$ with $j\Omega_p$ being the set of poles on the imaginary axis satisfying $\max_{\omega\in\Omega_p} |\omega|<\omega_c$,}
		and $C\in\mathcal{RH}_\infty^{n\times n}$ be frequency-wise sectorial.  Then $P\,\#\,C$ is stable if\\ i) for each $\omega\in[0,\omega_c)\setminus\Omega_p$, it holds $\bar{\phi}(P(j\omega))+\bar{\phi}(C(j\omega))<\pi$ and $\underline{\phi}(P(j\omega))+\underline{\phi}(C(j\omega))>-\pi$;\\
		ii) and for each $\omega\in[\omega_c,\infty]$,  it holds $\bar{\sigma}(P(j\omega))\bar{\sigma}(C(j\omega))<1. $
	\end{theorem}
	
	\begin{proof}
		Using a similar argument in \cite[Theorem~7]{wei2021phaseLTI}, we have that if $j\omega_0$ is a pole of $P$, it has to be simple. Recall the definition of contour $CT_\epsilon(j\Omega_p)$ in \eqref{eq:contour} with $\omega_{-\infty}=-\omega_c$ and $\omega_{\infty}=\omega_c$. 
		Noting the symmetry property of real-rational transfer matrices and the continuously defined phases of $P(s)$ and $C(s)$ along $CT_\epsilon(j\Omega_p)$, we obtain from condition~i) that for $s\in CT_\epsilon(j\Omega_p)$ with $\epsilon>0$ being sufficiently small, there exist continuous functions $\alpha(s),\beta(s)\in\mathbb{R}$ with $\alpha(\bar{s})=-\alpha(s)$, $\beta(\bar{s})=-\beta(s)$ and $\beta(s)-\alpha(s)\in(0,\pi)$ such that  $\underline{\phi}(P(s))>\alpha(s)$, $\underline{\phi}(C(s))>\pi-\alpha(s)$,  
		$\bar{\phi}(P(s))<\beta(s)$, $\bar{\phi}(C(s))<\pi-\beta(s)$; and obtain from condition~ii) that for $\omega\in[-\infty,-\omega_c]\cup[\omega_c,\infty]$ there exists $\gamma(j\omega)>0$ such that  $\bar{\sigma}(P(j\omega))<\gamma(j\omega)$ and $\bar{\sigma}(C(j\omega))<1/\gamma(j\omega)$. Let
		$$\delta(s):=\left\{\begin{matrix}1,~~s\in CT_\epsilon(j\Omega_p)\\
			0,~~\text{otherwise}\end{matrix}\right. ~~\text{and}~~\theta(s):=\left\{\begin{matrix}{\pi}/{2}-\beta(s),~~\beta(s)\geq-\alpha(s)\\
			{\pi}/{2}+\alpha(s),~~\beta(s)<-\alpha(s)\end{matrix}\right.,$$
		and construct the multiplier 
		\begin{align*}
			\Pi(s):=\delta(s)\begin{bmatrix}0 & e^{j\theta(s)}I_n\\ e^{-j\theta(s)}I_n & 0\end{bmatrix} + (1-\delta(s))\begin{bmatrix}\gamma(s)^2 I_n & 0\\0 & - I_n\end{bmatrix},~\forall~s\in CT^\infty_\epsilon(j\Omega_p).
		\end{align*}
		Moreover, based on the given conditions on $P$ and $C$, one can verify that  for all $s\in CT^\infty_\epsilon(j\Omega_p)$, 
		\begin{align}\label{eq:iqc_thm1}
			\begin{bmatrix} I_n\\\tau P(s) \end{bmatrix}^*\Pi(s)\begin{bmatrix} I_n\\\tau P(s) \end{bmatrix}\geq 0,~\forall~\tau\in[0,1],~~\text{and}~~
			\begin{bmatrix}-C(s) \\ I_n\end{bmatrix}^*\Pi(s)\begin{bmatrix} -C(s) \\ I_n\end{bmatrix}<0. 
		\end{align}
		It follows from \cite[Theorem~4.4]{Khong2016TAC} and \cite[Proposition~3]{Khong2018robust} that $P\cls C$ has no pole on or right to the contour $CT^\infty_\epsilon(j\Omega_p)$ for all any $\epsilon>0$. To show the stability of $P\cls C$, it then suffices to show that $P\cls C$ has no pole in $j\Omega_p$. To this end, suppose to the contrapositive that  $\omega_0\in\Omega_p$ is a pole of $P\cls C$.  Moreover, noting that for every $\omega\in[-\infty,\infty]$, $C(j\omega)$ is sectorial and thus invertible, we then obtain that there is no unstable pole-zero cancellation in the product of $P$ and $C$. 
		Therefore, $\omega_0\in\Omega_p$ is a pole of $P\cls C$ if and only if $\det(I+P(j\omega_0)C(j\omega_0)) =  0$.
		Due to \eqref{eq:iqc_thm1}, there exists an $\eta>0$ such that for all $\epsilon\geq 0$ being suffciently small and $s\in CT_\epsilon(j\Omega_p)$, it holds 
		\begin{align}\label{eq:pf_thm1_rotate}
			C(s)^{-1}e^{j\theta(s)}+C(s)^{-*}e^{-j\theta(s)}>\eta I~~\text{and}~~P(s)e^{j\theta(s)}+P(s)^*e^{-j\theta(s)} \geq 0.\end{align}
		Note that \begin{align*}0 &= \det(C(j\omega_0))\det(I+P(j\omega_0)C(j\omega_0))\\
			& =  \det(C(j\omega_0)^{-1}+P(j\omega_0)) = \det(e^{j\theta(j\omega_0)}C(j\omega_0)^{-1}+e^{j\theta(j\omega_0)}P(j\omega_0)) \\
			&= \det(C(j\omega_0)^{-1}e^{j\theta(j\omega_0)}+C(j\omega_0)^{-*}e^{-j\theta(j\omega_0)}+P(j\omega_0)e^{j\theta(j\omega_0)}+P(j\omega_0)^*e^{-j\theta(j\omega_0)}),
		\end{align*}
		where the last equality follows by the fact that for matrix $C$ satisfying $C+C^* \geq 0$, $\det(C)=0$ implies that $\det(C+C^*)=0$. 
		By continuity of transfer matrices, we have for $s\in SC_\epsilon(j\omega_0)$, 
		$\det(C(s)^{-1}e^{j\theta(s)}+C(s)^{-*}e^{-j\theta(s)}+P(s)e^{j\theta(s)}+P(s)^*e^{-j\theta(s)}) $ can be made arbitrarily small by taking $\epsilon>0$ to be sufficiently small. This contradicts to that $C(s)^{-1}e^{j\theta(s)}+C(s)^{-*}e^{-j\theta(s)}+P(s)e^{j\theta(s)}+P(s)^*e^{-j\theta(s)}>\eta I$ according to \eqref{eq:pf_thm1_rotate}. Therefore, $P\cls C$ has no pole in $j\Omega_p$ and the feedback stability is thus completely proved. 
	\end{proof}
	\begin{figure}
		\centering
		\includegraphics[width=.5\linewidth]{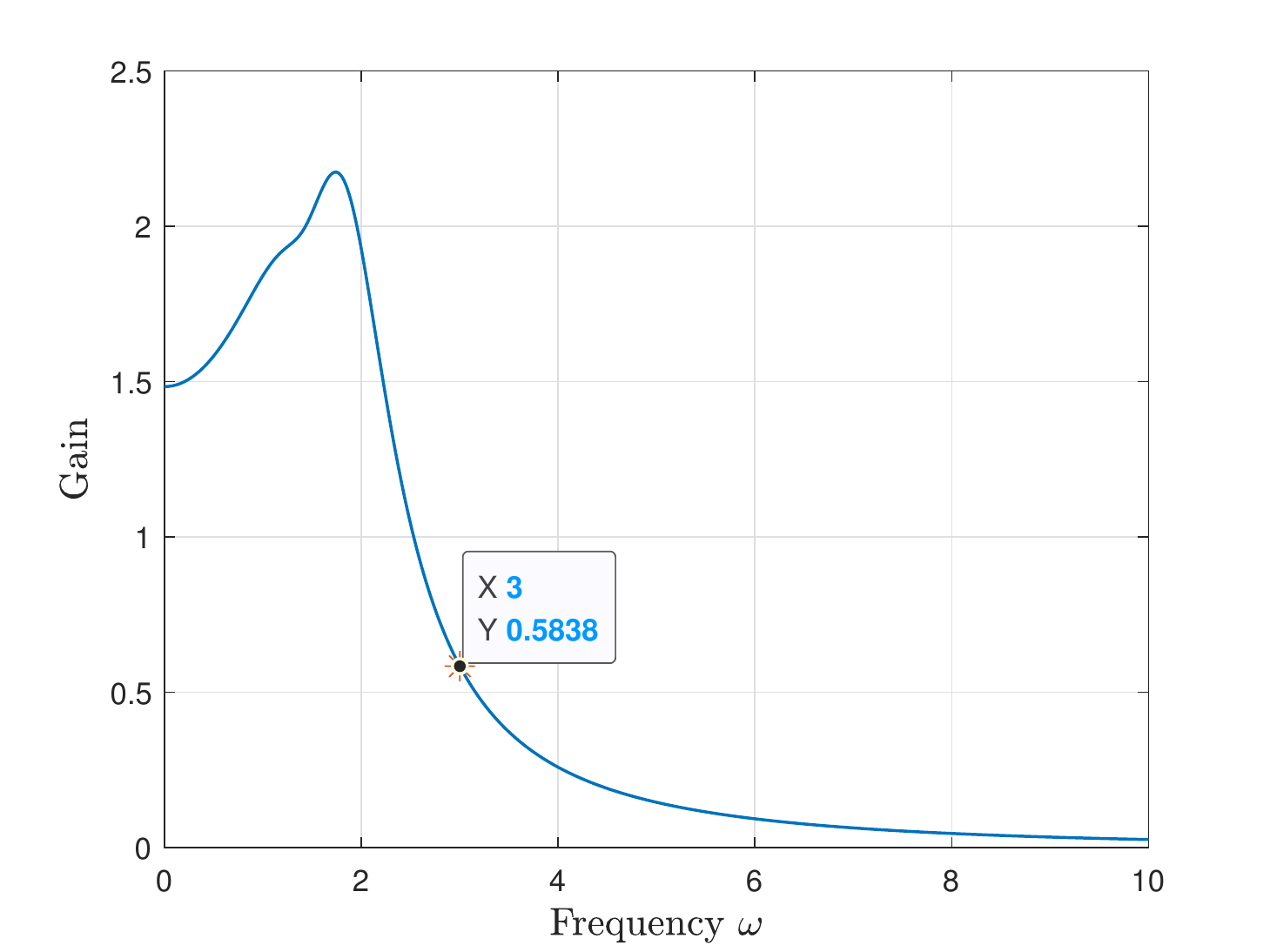}
		\caption{{The product gain $\bar{\sigma}(P(j\omega))\bar{\sigma}(C(j\omega))$ in terms of the frequency $\omega$. }}\label{fig:gain_product}
	\end{figure}
	\begin{figure}
		\centering
		\includegraphics[width=.5\linewidth]{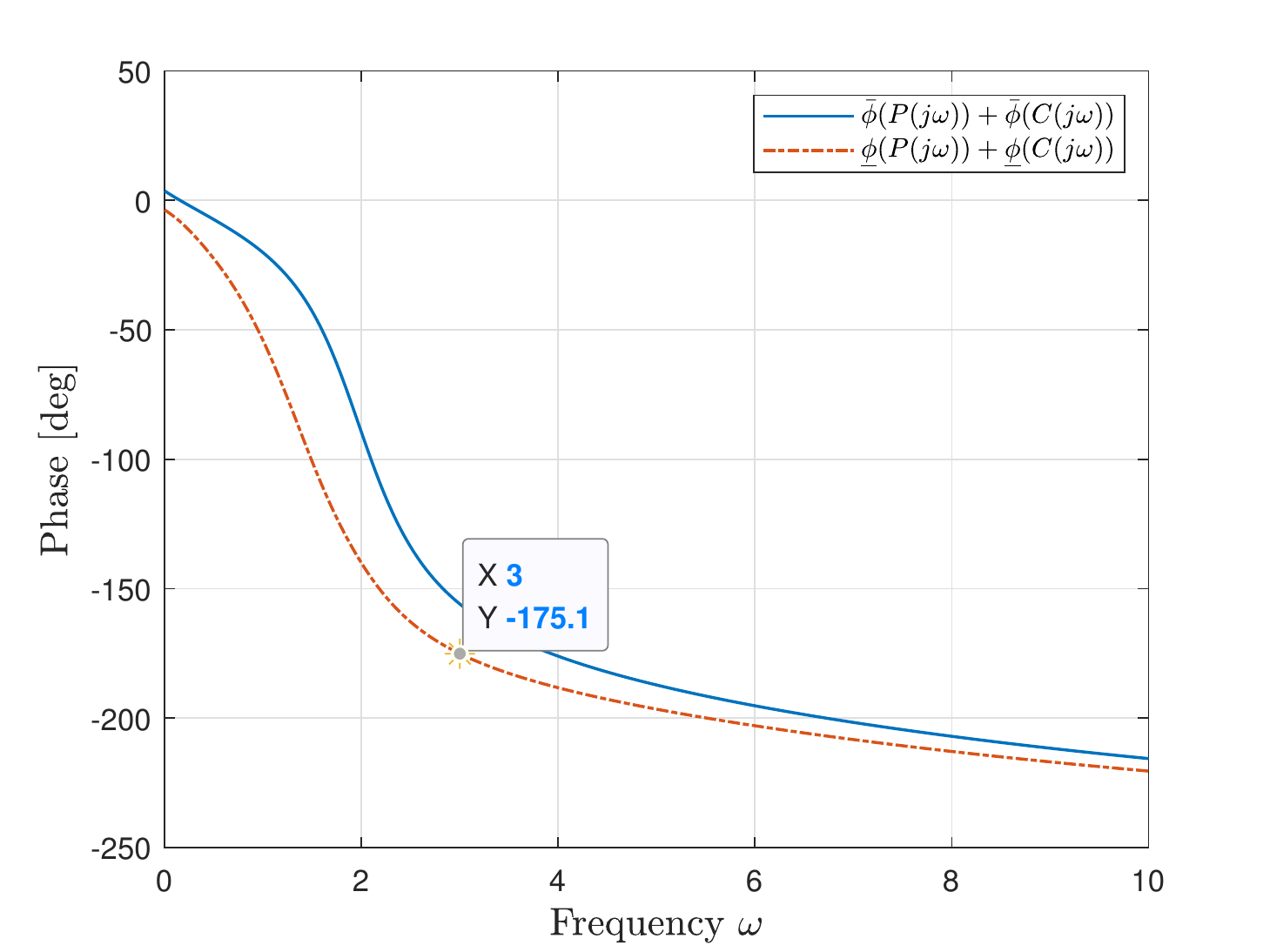}
		\caption{{The sums of phases:  $\bar{\phi}(P(j\omega))+\bar{\phi}(C(j\omega))$ and $\underline{\phi}(P(j\omega))+\underline{\phi}(C(j\omega))$ in terms of the frequency $\omega$.} }\label{fig:phase_sum}
	\end{figure}
	\begin{example}
		{Consider the following feedback control problem of a matrix second-order system, which have been intensively studied for engineering purposes \cite{Skelton1980AM,balas1982,gardiner1992stabilizing}. A matrix second-order system can be represented by the following transfer matrix:
			$$P(s) = (H_1s+H_2)(Ms^2+Cs+K)^{-1}B,$$
			where $M>0$ is the mass matrix, $C$ represents the sum of damping forces and gyroscopic efforts, $K\geq 0$ is the stiffness matrix, $B$ characterizes the input structure, and $H_1$ and $H_2$ characterizes the measured outputs. 
			Moreover, we observe that $P(0) = H_2K^{-1}B$ (if $K$ is nonsingular), $P(\infty)=0$, and  $\lim_{s\to\infty} sP(s) = H_1M^{-1}B$. We can easily discover that the gain of such a system become arbitrarily small at high frequency while its property at low frequency is roughly governed by the matrix $H_1M^{-1}B$ whose phases can be well estimated and even designed via a proper adjustment on the input matrix $B$ and output matrix $H_1$. For example, let the matrices take the following values 
			$$M=B=I_3,~C=\begin{bmatrix}
				3 & 0 & 0\\
				0 & 2 & 0\\
				0 & 1 & 2
			\end{bmatrix},~K = \begin{bmatrix}
				6 & 0 & 2\\
				0 & 7 & 0\\
				2 & 1 & 7
			\end{bmatrix},~H_1=\frac{1}{100}\begin{bmatrix}
				3 & 2 & 1\\
				1 & 3 & 0\\
				0 & 1 & 2
			\end{bmatrix},~H_2=\begin{bmatrix}
				70 & 0 & 2\\
				0 & 70 & 1\\
				0 & 2 & 60
			\end{bmatrix}.$$
			Such a matrix second-order system can be stabilized by a simple first-order diagonal controller $$C(s) = \frac{1}{s+10}I_3. $$
			To obtain the feedback stability, we can take advantages of the proposed mixed gain-phase stability result with a cut-off frequency in Theorem~\ref{thm:mixture}. As shown in Fig.~\ref{fig:gain_product}, the product of gains of $P(j\omega)$ and $C(j\omega)$ is large (greater than 1) in the low frequency, leading to the failure of a pure small gain theorem for the feedback stability. On the other hand, as shown in Fig.~\ref{fig:phase_sum}, the sum of the smallest phases of $P(j\omega)$ and $C(j\omega)$ becomes smaller than $-\pi$ (-180 degrees) at high frequency range, rendering a pure small phase statement to fail. Nevertheless, as revealed by Fig.~\ref{fig:gain_product} and \ref{fig:phase_sum},
			an application of Theorem~\ref{thm:mixture} directly shows the feedback stability of $P\cls C$ when the cut-off frequency is set to be $\omega_c=3$. }
	\end{example}
	
	%

	\subsection{A Frequency-wise Mixed Gain-Phase Stability Criterion}
	Obviously, the ways to combine the gain and phase properties in the analysis of feedback stability are not unique. In the following, we explore more general mixed gain and phase stability results by extending the matrix results in Lemmas~\ref{lem:mixGP_matrix} and \ref{lem:mixGP_matrix_nec} as well as the system result in Theorem~\ref{thm:mixture}.

	\begin{theorem}\label{thm:smallvase}
		Let $P$ be {semi-stable frequency-wise semi-sectorial with $j\Omega_p$ being the set of poles on the imaginary axis} and $C\in\mathcal{RH}_\infty$. Let $\gamma(\omega)\in(0,\infty)$ and $\alpha(\omega),\beta(\omega)\in\mathbb{R}$ with $\beta(\omega)-\alpha(\omega)\in(0,\pi]$, $\forall~\omega\in[0, \infty]$, be piece-wise continuous functions.  Then feedback system $P\,\#\,C$ is stable if \\ i) for each $\omega\in[0,\infty]\setminus\Omega_p$, it holds $P(j\omega)\in\bar{{\mathcal P}}(\alpha(\omega),\beta(\omega))\cup\bar{{\mathcal G}}(\gamma(\omega))$;\\
		ii) and for each $\omega\in[0,\infty]$,  it holds $C(j\omega)\in{{\mathcal P}}(-\pi-\alpha(\omega),\pi-\beta(\omega))\cap{{\mathcal G}}(\gamma(\omega)^{-1})$. 
		
	\end{theorem}
	\begin{proof}
		By continuity of gains and phases of $P$ and $C$ along contours, for $\epsilon>0$ being sufficiently small, there exist piece-wise continuous functions $\tilde{\alpha}(s)$, $\tilde{\beta}(s)$ and $\tilde{\gamma}(s)$, $s\in CT_\epsilon^\infty(j\Omega_p)$, which satisfy that $P(s)\in\bar{{\mathcal P}}(\tilde{\alpha}(s),\tilde{\beta}(s))\cup\bar{{\mathcal G}}(\tilde{\gamma}(s))$ and $C(s)\in{{\mathcal P}}(-\pi-\tilde{\alpha}(s),\pi-\tilde{\beta}(s))\cap{{\mathcal G}}(\tilde{\gamma}(s)^{-1})$.
		Construct for each $s\in CT_\epsilon^\infty(j\Omega_p)$ that
		$$\delta(s):=\left\{\begin{array}{l}1,~~P(s)\in\bar{{\mathcal P}}(\alpha(s),\beta(s))\\
			0,~~\text{otherwise}~~~~~~~~~\end{array}\right.,~~\theta(s):=\left\{\begin{matrix}{\pi}/{2}-\beta(s),~~\beta(s)\geq-\alpha(s)\\
			{\pi}/{2}+\alpha(s),~~\beta(s)<-\alpha(s)\end{matrix}\right.,$$
		and
		\begin{align*}
			\Pi(s):=\delta(s)\begin{bmatrix}0 & e^{j\theta(s)}I_n\\ e^{-j\theta(s)}I_n & 0\end{bmatrix} + (1-\delta(s))\begin{bmatrix}I_n & 0\\0 & -\gamma(s)^2 I_n\end{bmatrix}.
		\end{align*}
		Based on the given conditions on $P$ and $C$, one can verify that for all $s\in CT^\infty_\epsilon(j\Omega_p)$, 
		\begin{align}\label{eq:iqc_thm2}
			\begin{bmatrix} I_n\\\tau P(s) \end{bmatrix}^*\Pi(s)\begin{bmatrix} I_n\\\tau P(s) \end{bmatrix}\geq 0,~\forall~\tau\in[0,1],~~\text{and}~~
			\begin{bmatrix}-C(s) \\ I_n\end{bmatrix}^*\Pi(s)\begin{bmatrix} -C(s) \\ I_n\end{bmatrix}<0. 
		\end{align}
		It follows from \cite[Theorem~4.4]{Khong2016TAC} and \cite[Proposition~3]{Khong2018robust} that $P\cls C$ has no pole on or right to the contour $CT^\infty_\epsilon(j\Omega_p)$ for all sufficiently small $\epsilon>0$.
		The rest of the proof then follows by similar arguments as those in the proof for Theorem~\ref{thm:mixture}. 
	\end{proof}
	Theorem~\ref{thm:smallvase} reduces to Theorem~\ref{thm:mixture}  if we properly restrict the position of the imaginary-axis poles of $P$ and specifically take
	\begin{align*}
		(\alpha(\omega),\beta(\omega))=\hspace{-3pt}\left\{\begin{matrix}(\alpha(\omega),\beta(\omega)), & \hspace{-5pt}\omega\in[0,\omega_c)\\ \pi, & \hspace{-5pt}\omega \in[\omega_c,\infty] \end{matrix}\right.~\text{and}~\gamma(\omega)=\hspace{-3pt}\left\{\begin{matrix}\infty, & \hspace{-5pt}\omega\in[0,\omega_c)\\ \gamma(\omega), & \hspace{-5pt}\omega \in[\omega_c,\infty] \end{matrix}\right. \hspace{-1pt}.
	\end{align*}
	
	The following theorem develops a necessary and sufficient robust stability result when the plant $P$ is within a vase-shape uncertainty set with a $\pi$-phase spread. The matrix version of the theorem can be found in \cite{WANG2020PhaseMath}. 
	\begin{theorem}[A small Vase Theorem with Necessity]\label{thm:smallvase_Nec}
		Let $C\in\mathcal{RH}_\infty$ be frequency-wise sectorial, $g,g^{-1}\in\mathcal{RH}_\infty$ with $|g(j\omega)|\in(0,\infty)$, $\omega\in[-\infty,\infty]$, and $h,h^{-1}\in\mathcal{RH}_\infty$ with $\angle h(j\omega)\in(-\pi/2,\pi/2)$, $\omega\in[-\infty,\infty]$. Then $P\,\#\,C$ is stable for all $P\in\mathcal{RH}_\infty$ with $P(j\omega)\in\bar{{\mathcal P}}(-\pi/2+\angle h(j\omega), \pi/2+\angle h(j\omega))\cup\bar{{\mathcal G}}(|g(j\omega)|)$, $\omega\in[0, \infty]$ if and only if
		\vspace{-5pt}
		\begin{align}\label{eq:thm_stability_w_necessity}
			C(j\omega)\in{{\mathcal P}}\Big(-\pi/2-\angle h(j\omega),\pi/2-\angle h(j\omega)\Big)\cap{{\mathcal G}}\Big(|g(j\omega)|^{-1}\Big),~\omega\in[0, \infty].\end{align}
	\end{theorem}
	\begin{proof}
		The sufficiency follows from Theorem~\ref{thm:smallvase} by letting $\alpha(\omega)=-\pi/2+\angle h(j\omega)$, $\beta(\omega)=\pi/2+\angle h(j\omega)$, and $\gamma(\omega)=|g(j\omega)|$.
		
		As for the necessity, we suppose to the contraposition that $C$ does not satisfy \eqref{eq:thm_stability_w_necessity}. In this case, there exists an $\omega_0\in[-\infty,\infty]$ such that (i) $\bar{\sigma}(C(j\omega_0)) \geq \gamma(\omega_0)^{-1}$ or (ii) $\bar{\phi}(C(j\omega_0)) \geq \pi/2+\angle h(j\omega)$ or (iii) $\underline{\phi}(C(j\omega_0)) \leq -\pi/2+\angle h(j\omega)$. 
		
		In case (i), we obtain from \cite[Theorem~8.1]{Zhou1998Essential} that there exists $\tilde{P}$ with $\tilde{P}(j\omega)\in\bar{{\mathcal G}}(\gamma(\omega))$, $\omega\in[-\infty, \infty]$, such that $\tilde{P}\,\#\,C$ is unstable. In cases (ii) or (iii), we obtain from \cite[Theorem~2]{wei2021phaseLTI} that there exists $\tilde{P}$ with $\tilde{P}(j\omega)\in\bar{{\mathcal P}}(-\pi/2+\angle h(j\omega), \pi/2+\angle h(j\omega))$, $\omega\in[-\infty, \infty]$, such that $\tilde{P}\,\#\,C$ is unstable. By contraposition, we show the necessity. 
	\end{proof}
	
	In contrast to the above theorem, one may wonder naturally if there is a parallel result when the uncertainty set is in a fan-shape region. To be precise, we hope to find a mixed gain and phase characterization for $C\in\mathcal{RH}_\infty$ so that  $P\,\#\,C$ is stable for all $P\in\mathcal{RH}_\infty$ with $P(j\omega)\in{{\mathcal P}}(-\pi/2+\angle h(j\omega), \pi/2+\angle h(j\omega))\cap{{\mathcal G}}(|g(j\omega)|)$, $\omega\in[-\infty, \infty]$. Taking the set for $C$ as $\{C\in\mathcal{RH}_\infty:~C(j\omega)\in\bar{{\mathcal P}}(-\pi/2+\angle h(j\omega), \pi/2+\angle h(j\omega))\cup\bar{{\mathcal G}}(|g(j\omega)|)$, $\omega\in[-\infty, \infty]$\} results in a sufficient robust stability condition for $P\cls C$ in light of Theorem~\ref{thm:smallvase}, while a necessary and sufficient condition is still undisclosed. Finding the necessary and sufficient condition is challenging, and we call it a ``half-disk'' problem, as a vivid illustration on how we characterize the set of uncertainty as in Fig.~\ref{fig:fan}. We believe it is still an open problem, even for the matrix case.

	\vspace{5pt}

	\subsection{Application: Robust Stabilization}\label{sec:simu_stab}
	{Given matrix $K\in\mathbb{R}^{n\times n}$, consider a set of matrices $\mathcal{S}_{\delta,\eta}(K) :=\{A\in\mathbb{R}^{n\times n}:~KA~\text{is accretive},~\underline{\sigma}(A)\geq \delta,~\bar{\sigma}(A)\leq \eta\}$ and a set of stable systems $\mathcal{B}_\gamma:=\{P\in\mathcal{RH}_\infty^{n\times n}:~\|P\|_\infty\leq \gamma\}$. 
		\begin{theorem}
			For $\gamma,\delta,\eta>0$ and $K\in\mathbb{R}^{n\times n}$ being nonsingular, the following set of systems $$\mathcal{M}:=\left\{\frac{A}{s}+P(s):~A\in\mathcal{S}_{\delta,\eta}(K),~P\in \mathcal{B}_\gamma\right\}$$ can be robustly stabilized by $\epsilon K$ with any $\epsilon>0$ being sufficiently small. 
		\end{theorem}
		
		\begin{proof} 
			For any $C\in\mathbb{C}^{n\times n}$, denote by $\bar{\kappa}(C):=\max\{|x^*Cx|:~x\in\mathbb{C}^n,~|x|=1\}$ the numerical radius of $C$ and by $\underline{\kappa}(C):=\min\{|x^*Cx|:~x\in\mathbb{C}^n,~|x|=1\}$ the least distance from the numerical range of $C$ to the origin. It is clear that $\|C\| = \bar{\sigma}(C) \geq \bar{\kappa}(C) \geq \underline{\kappa}(C) \geq \underline{\sigma}(C)$. 
			
			Since $KA$ is accretive and $\underline{\sigma}(A)\geq \delta$, it holds for all $C$ with $\|C\|<\delta\underline{\sigma}(K)$ that 
			$$\underline{\kappa}(KA+C)\geq\underline{\sigma}(KA+C) > \underline{\sigma}(K) \underline{\sigma}(A)-\|C\| =  \underline{\sigma}(K)(\underline{\sigma}(A)-\delta)\geq 0,$$
			whereby $KA+C$ is sectorial. Since $KA+C$ is sectorial for all $C$ with $\|C\|<\delta\underline{\sigma}(K)$ and $KA$ is accretive, by continuity there exists a $\tilde{\delta}>0$ such that for all $\tilde{C}$ with $\|\tilde{C}\|<\tilde{\delta}$,  $KA+\tilde{C}$ is acrretive. 
			Together with $P\in \mathcal{B}_\gamma$, we have that
			$\frac{KA}{j\omega}+KP(j\omega)$ is accretive for all $\omega\in(0,\frac{\tilde{\delta}}{\gamma\bar{\sigma}(K)})$. Let $\omega_c\in(0,\frac{\tilde{\delta}}{\gamma\bar{\sigma}(K)})$. On the other hand, as $\bar{\sigma}(KA)\leq\bar{\sigma}(K)\bar{\sigma}(A)\leq \eta\bar{\sigma}(K)$, it holds that
			$$\left\|\frac{KA}{j\omega}+KP(j\omega)\right\|\leq \frac{\eta\bar{\sigma}(K)}{\omega_c}+\gamma\bar{\sigma}(K)=:c,~\forall~\omega\in[\omega_c,\infty].$$
			Therefore, invoking Theorem~\ref{thm:mixture} on feedback interconnection between $\frac{KA}{s}+KP(s)$ and $\epsilon I$ with $\epsilon\in(0,1/c)$ and cutoff frequency $\omega_c$, we obtain the feedback stability of $(\frac{KA}{s}+KP(s))\cls (\epsilon I)$ as well as that of $(\frac{A}{s}+P(s))\cls (\epsilon K)$. 
		\end{proof}
	}
	\section{Mixed Gain/Phase Stability Condition by Davis-Wielandt Shell}\label{sec:DW shell}
	In this section, we first try to recover the small phase and small gain theorems using the Davis-Wielandt (DW) shell. Afterwards, we explore a mixed gain/phase stability result based on the shell. 
	
	The Davis-Wielandt shell \cite{davis1968shell,lestas2012large}, defined by $$\mathcal{DW}(A):=\{(\text{Re}~x^*Ax,\text{Im}~x^*Ax,\|Ax\|^2):~x\in\mathbb{C}^n,\|x\|=1\},$$ 
	is a higher dimensional generalization of the numerical range. 
	\subsection{Vertically Projected DW-Shell and Constrained Phases}
	First of all, we extend the definition of sectorial matrices using the notion of DW-shell. For $A\in\mathbb{C}^{n\times n}$ and $r\in\mathbb{R}$, define the vertically projected (or, v-projected) DW-shell by
	$${\mathcal{W}_{\geq r}(A)}:=\left\{(p,q):~(p,q,h^2)\in\mathcal{DW}(A),~h\geq r\right\}.$$
	It is named in this way because $\mathcal{W}_{\geq r}(A)$ is essentially the projection of the components in the whole shell $\mathcal{DW}(A)$ above the hyperplane $z=r^2$ in $\mathbb{R}^3=\{(x,y,z)\}$ onto the hyperplane. It is noteworthy that $\mathcal{W}_{\geq r}(A) = \mathcal{W}(A)$ when $r\leq \underline{\sigma}(A)$, revealing the fact that DW-shell is a generalization of numerical range. 
	
	\begin{proposition}\label{prop:compact_conv}
		The v-projected DW-shell $\mathcal{W}_{\geq r}(A)$ is compact and convex.
	\end{proposition}
	\begin{proof}
		The case when $r>\bar{\sigma}(A)$ is trivial since the projected shell becomes an empty set. In what follows, we consider the case when $r\in[0,\bar{\sigma}(A)]$. Denote by $X:=\{x\in\mathbb{C}^n:~\|Ax\| \geq r,~\|x\|=1\}$, which is non-empty and compact. Then it can be verified that
		\begin{align}\label{eq:pf_prop1}
			\{(\Re x^*Ax,\Im x^*Ax):~x\in X\}=\mathcal{W}_{\geq r}(A).
		\end{align}
		
		Next we show that the projected DW-shell $\mathcal{W}_{\geq r}(A)$ is compact. Clearly it is bounded. Let $\{f_i=(\Re x_i^*Ax_i,\Im x_i^*Ax_i)\}_{i=1}^\infty\subset\mathcal{W}_{\geq r}(A)$ be an arbitrary Cauchy sequence. Note that each $x_i$ belongs to $X$, which is a compact set, whereby $x=\lim_{i\to\infty} x_i \in X$. Using \eqref{eq:pf_prop1}, we obtain that
		$$\lim_{i\to\infty}f_i=\lim_{i\to\infty}(\Re x_i^*Ax_i,\Im x_i^*Ax_i)
		\in \{(\Re x^*Ax,\Im x^*Ax):~x\in X\}=\mathcal{W}_{\geq r}(A),
		$$
		yielding that the projected DW-shell is compact.
		
		Lastly, we show its convexity. By \cite[Theorem~2]{binding1985}, we know $\mathcal{DW}(A)$ is convex when $n\neq 2$, whereby $\{(p,q,h^2)\in\mathcal{DW}(A):~h\geq r\}:=S$ is convex. Since $\mathcal{W}_{\geq r}(A)$ is the projection of $S$ onto the hyperplane $z=r^2$ in $\mathbb{R}^3=\{(x,y,z)\}$, it is convex as well. On the other hand, when $n=2$, we obtain from \cite[Theorem~2]{binding1985} that $\mathcal{DW}(A)$ is an affine sphere. Clearly, $\mathcal{W}_{\geq r}(A)$, as the projection of the upper part ($z\geq r^2$) of the affine sphere $\mathcal{DW}(A)$ onto the hyperplane $z=r^2$, should be convex.
	\end{proof}
	
	\begin{definition}[Constrained Sectorial Matrices]
		We say a matrix $A\in\mathbb{C}^{n\times n}$ is $r$-sectorial if $$(0,0)\notin \mathcal{W}_{\geq r}(A).$$
	\end{definition}
	In particular, when $r=0$, $0$-sectorial reduces to our usual definition for sectorial matrices. When $r>\bar{\sigma}(A)$, $\mathcal{W}_{\geq r}(A)$ is an empty set. 
	
	Since the projected DW-shell $\mathcal{W}_{\geq r}(A)$ is compact and convex, $(0,0)$ being outside the set implies that $\mathcal{W}_{\geq r}(A)$ is located in a certain open half-plane on $\mathbb{R}^2$. Consequently, a phase sector can be defined to measure the phase spread of $\mathcal{W}_{\geq r}(A)$, as detailed in the following.
	\begin{definition}[Constrained Phase Sector]
		For $A\in\mathbb{C}^{n\times n}$ and $r\in[\underline{\sigma}(A),\bar{\sigma}(A)]$, the constrained phase sector of an $r$-sectorial matrix $A$ is defined as $\Psi_r(A):=[\underline{\psi}_r(A),\bar{\psi}_r(A)]$ with
		\begin{align*}
			&\underline{\psi}_r(A):=\min_{z} \{\angle z:~(\Re z, \Im z)\in\mathcal{W}_{\geq r}(A)\},\\
			&\bar{\psi}_r(A):=\max_{z} \{\angle z:~(\Re z, \Im z)\in\mathcal{W}_{\geq r}(A)\}
		\end{align*}
		satisfying that
		$$\frac{\underline{\psi}_r(A)+\bar{\psi}_r(A)}{2}\in(-\pi,\pi]
		~~\text{and}~~
		\bar{\psi}_r(A)-\underline{\psi}_r(A)\in(0,\pi].$$
		For notational convenience, for $r >\bar{\sigma}(A)$, we let $\underline{\psi}_r(A)=+\infty$ and $\bar{\psi}_r(A) = -\infty$. 
	\end{definition}

	The constrained phases possess the following monotone property. 
	\begin{proposition}\label{prop:monotone_dw_phase}
		Let $r\in\mathbb{R}$, $\tau\in(0,1]$ and $A\in\mathbb{C}^{n\times n}$ be $r$-sectorial. Then the following statements are true. 
		{\begin{itemize}
				\item [(a)] $\underline{\psi}_r(A)$ is an increasing function of $r$ while $\bar{\psi}_r(A)$ is a decreasing one. 
				\item [(b)] 
				$\underline{\psi}_r(\tau A)$ is a decreasing function of $\tau$ while $\bar{\psi}_r(\tau A)$ is an increasing one. 
		\end{itemize}}
	\end{proposition}
	\begin{proof}
		The statement~(a) follows from the set inclusion:
		$$r\leq h~~\Rightarrow~~\mathcal{W}_{\geq r}(A)\supset\mathcal{W}_{\geq h}(A). $$
		
		As for statement~(b), for any $x\in\mathbb{C}^n$ with unit length, we have 
		$$\|\tau Ax\| \geq r ~~\Leftrightarrow~~\|Ax\| \geq r/\sqrt{\tau},$$
		whereby 
		$$\bar{\psi}_r(\tau A) = \bar{\psi}_{r/\sqrt{\tau}}(A)~~\text{and}~~\underline{\psi}_r(\tau A) = \underline{\psi}_{r/\sqrt{\tau}}(A).$$
		Noting that $r/\sqrt{\tau}$ is an decreasing function of $\tau$, we therefore show statement~(b) using statement (a) and the composition rules of monotone functions.  
	\end{proof}
	The constrained phase sector can be computed via a semi-definite programming (SDP) as revealed in the following proposition. 
	{\begin{proposition}\label{prop:compute_constrained_phase}
			Suppose that $A\in\mathbb{C}^{n\times n}$ is $r$-sectorial with $\Psi_r(A)\subset(-\pi/2,\pi/2)$. Let its Hermitian and skew-Hermitian parts be $A_h=(A+A^*)/2$ and $A_s=(A-A^*)/(2j)$, respectively.  Then $\underline{\psi}_r(A)=\arctan(g^\star)$ and $\bar{\psi}_r(A)=\arctan(h^\star)$, where $g^\star$ and $h^\star$ are respectively the optimal values obtained by solving the following SDP problems:
			\begin{align}
				&\label{eq:sdp1}\max_{g,\tau}~ g,~~\text{\rm s.t.}~~ gA_h-A_s+\tau(A^*A-r^2 I) \leq 0,~\tau\geq 0;~\text{and}\\
				&\label{eq:sdp2}\min_{h,\tau}~ h,~~\text{\rm s.t.}~~ A_s-hA_h+\tau(A^*A-r^2 I) \leq 0,~\tau\geq 0.
			\end{align}
		\end{proposition}
		It is noteworthy that when $\Psi_r(A)$ does not belong to $(-\pi/2,\pi/2)$, a proper rotation of $A$ in advance will make it happen. 
		\begin{proof}
			In what follows, we only prove that $\bar{\psi}_r(A)=\arctan(h^\star)$, while the other can be derived similarly. Since $\Psi_r(A)\subset(-\pi/2,\pi/2)$, $\tan\bar{\psi}_r(A)$ is the smallest gradient of line tangent to $\mathcal{W}_{r\geq 0}(A)$ from above and crossing the origin. As a result, $\tan\bar{\psi}_r(A)$ is the solution to the following problem:
			$$\min_{h}~h,~~\text{s.t.}~~\Im(x^*Ax)\leq h\Re(x^*Ax),~\forall~x\in\mathbb{C}^n~\text{with}~\|x\|=1,~\|Ax\|\geq r,$$
			which is equivalent to that
			\begin{align}\label{eq:pf_sdp}
				\min_{h}~h,~~\text{s.t.}~~x^*(A_s-hA_h)x\leq 0,~\forall~x\in\mathbb{C}^n~\text{with}~x^*(A^*A-r^2I)x\geq 0.
			\end{align}
			By the S-procedure theory \cite{Hara2005GeneralizedKYP}\cite{yakubovich1971}, the statement $x^*(A_s-hA_h)x\leq 0$ holds for all $x\in\mathbb{C}^n$ with $x^*(A^*A-r^2I)x\geq 0$ is equivalent to that there exists a $\tau\geq 0$ such that $A_s-hA_h+\tau(A^*A-r^2 I) \leq 0$, whereby the problem~\eqref{eq:pf_sdp} is further equivalent to \eqref{eq:sdp2}. To conclude we have that  $\tan\bar{\psi}_r(A)$ is the optimal value to \eqref{eq:sdp2}, which completes the proof. 
		\end{proof}
		\begin{figure}[H]
			\centering
			\includegraphics[width=.5\linewidth]{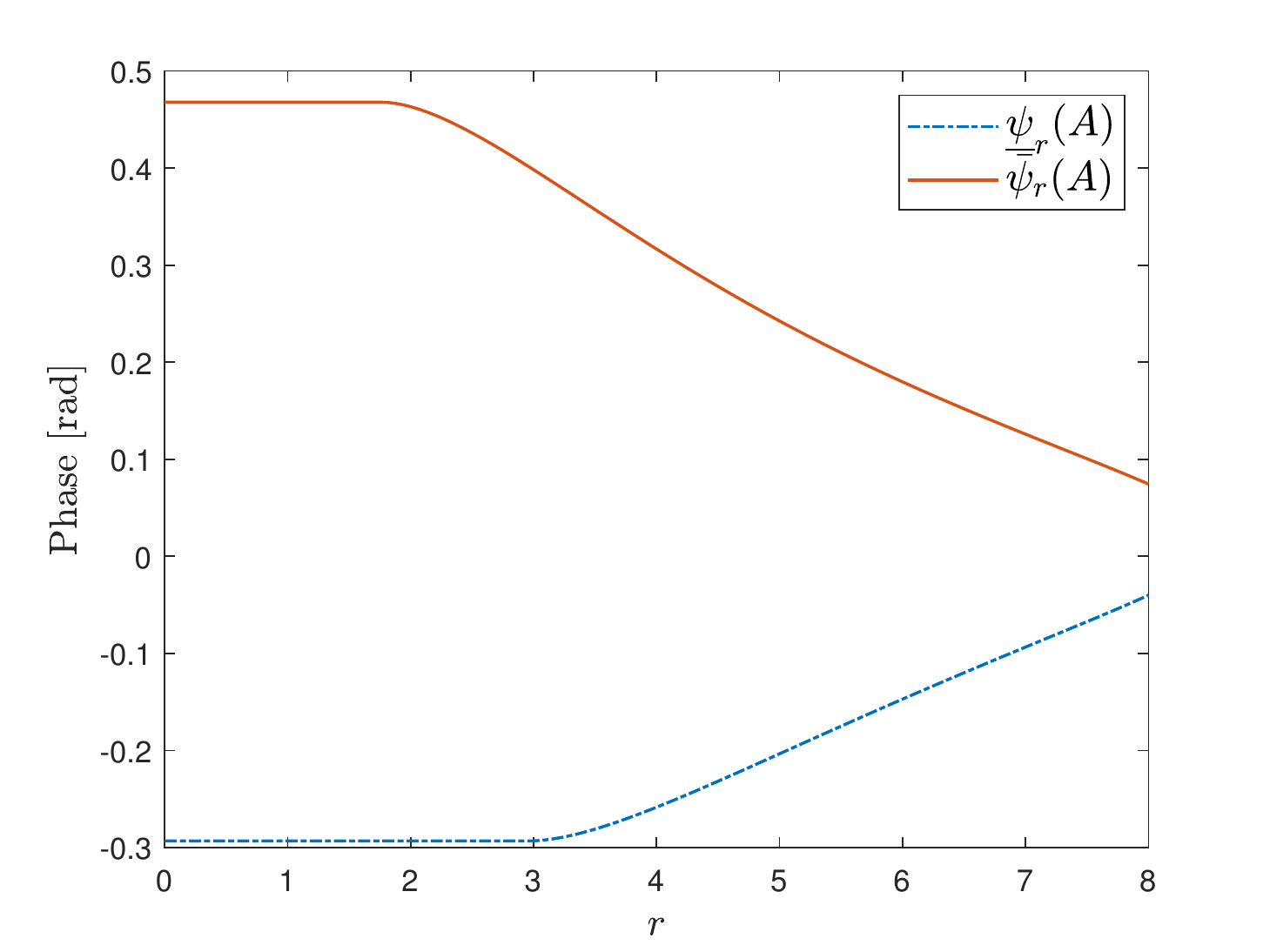}
			\caption{{$\underline{\psi}_r(A)$ and $\bar{\psi}_r(A)$ in terms of $r$.} }\label{fig:example_constrained_phase}
		\end{figure}
		Consider, for example, the following complex valued matrix
		$$A = \begin{bmatrix}
			5 & 2 & 1+j\\
			3 &6 &2\\
			0 &2 &2
		\end{bmatrix}. $$
		Its constrained phase sector $\Psi_r(A) = [\underline{\psi}_r(A),\bar{\psi}_r(A)]$ is computed using Proposition~\ref{prop:compute_constrained_phase} and drawn in terms of a continuous variation of $r$ from 0 to 8, as shown in Fig.~\ref{fig:example_constrained_phase}. 
	}
	
	The following result concerns the invertibility of matrix $I+AB$, which will help establish feedback stability for LTI systems. 
	\begin{theorem}\label{thm:DW_matrix}
		Let $A,B\in\mathbb{C}^{n\times n}$ and $B$ be sectorial. Then $I+AB$ is invertible if there exists an $r \geq 0$ such that $A$ is $r$-sectorial and
		\begin{align}\label{eq:thm_vp}
			&\bar{\psi}_r(A)+\bar{\phi}(B)<\pi,~\underline{\psi}_r(A)+\underline{\phi}(B)>-\pi,~\bar{\sigma}(B)r<1.
		\end{align}
	\end{theorem}
	It is noteworthy that when $r=0$, the theorem reduces to the small phase theorem for matrices in \cite{WANG2020PhaseMath}.
	
	\begin{proof}[Proof of Theorem~\ref{thm:DW_matrix}]
		It follows from the first two inequalities in \eqref{eq:thm_vp} that there exists an $\alpha\in(-\pi,\pi)$ such that both $\Psi_r(A)+\alpha$ and $\Psi(B)-\alpha$ are contained in $(-\pi/2,\pi/2)$. Let
		$$
		\Pi_p=\begin{bmatrix}0 & e^{-j\alpha}I_n\\ e^{j\alpha}I_n & 0\end{bmatrix}
		~\text{and}~\Pi_g=\begin{bmatrix}I_n & 0\\  0 & -r^2I_n\end{bmatrix}.
		$$
		Note that
		$$\{x\in\mathbb{C}^n:~\|Ax\| \geq r,~\|x\|_2=1\}=\left\{x\in\mathbb{C}^n:~x^*\begin{bmatrix}A \\ I\end{bmatrix}^*\Pi_g\begin{bmatrix}A \\ I\end{bmatrix}x \geq 0\right\}:=X.$$
		By the definition of the constrained phases and the fact that $\Psi_r(A)+\alpha\subset(-\pi/2,\pi/2)$, we obtain that
		$$x^*\begin{bmatrix}A \\ I\end{bmatrix}^*\Pi_p\begin{bmatrix}A \\ I\end{bmatrix}x<0,$$
		for all $x\in X$. It then follows from the S-procedure theory \cite{Hara2005GeneralizedKYP}\cite{yakubovich1971} that there is a $\tau\geq 0$ such that 
		$$\begin{bmatrix}A \\ I\end{bmatrix}^*(\Pi_p+\tau\Pi_g)\begin{bmatrix}A \\ I\end{bmatrix} \leq 0.$$
		On the other hand, noting that $\Psi(B)-\alpha\subset(-\pi/2,\pi/2)$, we obtain that
		\begin{align}\label{eq:pf_DW_mixgp1}
			\begin{bmatrix}I \\ -B\end{bmatrix}^*\Pi_p\begin{bmatrix}I\\ -B\end{bmatrix}>0. \end{align}
		From $\bar{\sigma}(B)r<1$ and $\tau\geq 0$, we have 
		$$\begin{bmatrix}I \\ -B\end{bmatrix}^*\tau\Pi_g\begin{bmatrix}I\\ -B\end{bmatrix}\geq 0.$$
		Summing up the above inequalities yields that
		\begin{align}\label{eq:pf_DW_mixgp2}
			\begin{bmatrix}I \\ -B\end{bmatrix}^*(\Pi_p+\tau\Pi_g)\begin{bmatrix}I\\ -B\end{bmatrix}> 0.\end{align}
		Combining \eqref{eq:pf_DW_mixgp1} and \eqref{eq:pf_DW_mixgp2}, we obtain that $I+AB$ is invertible by the IQC stability condition \cite[Theorem~1]{Megretski1997IQC} adapted for matrices.
	\end{proof}
	
	\subsection{Horizontally Projected DW-Shell and Constrained Gains}
	A dual problem arises naturally that instead of fixing a height $r$ in the DW-shell, we fix a particular phase and study the gain properties of a matrix via its DW-shell. Such a case can be viewed as a DW-shell being projected horizontally onto some vertical lines in $\mathbb{R}^3$.
	
	For a sectorial matrix $A\in\mathbb{C}^{n\times n}$ and $\theta\in[0,\pi)\cap \Psi(A)$, define the constrained gain range by
	$${\mathcal{R}_{\geq \theta}(A)}:=\{h\geq 0:~z\in\mathbb{C},~\angle z\notin(-\theta,\theta),~(\Re z,\Im z,h^2)\in\mathcal{DW}(A)\}.$$
	Clearly, for any $h\in\mathcal{R}_{\geq \theta}(A)$, it holds that $h\in[\underline{\sigma}(A),\bar{\sigma}(A)]$, whereby this set $\mathcal{R}_{\geq \theta}(A)$ contains partial information about the gains of $A$ under phase constraints, or more precisely, when the phases are large. This is a parallel notion to the constrained phase sector introduced in the previous subsection, as we simply swap the roles between gain and phase. 
	Using similar arguments in Proposition~\ref{prop:compact_conv}, we obtain that $\mathcal{R}_{\geq \theta}(A)$ is a bounded set in $\mathbb{R}$. The maximum value of $\mathcal{R}_{\geq \theta}(A)$ is denoted as 
	$\gamma_\theta(A):=\max \{h\in\mathcal{R}_{\geq \theta}(A)\}$, which characterizes the largest gain information of $A$ under certain phase constraint. 
	
	{Similarly to the computation of the constrained phases using Proposition~\ref{prop:compute_constrained_phase}, the constrained gain can be obtained by solving a pair of SDP problems as well, shown in the following result. 
		\begin{proposition}
			Let $A\in\mathbb{C}^{n\times n}$ be accretive and $\theta\in[0,\pi/2)$. Let its Hermitian and skew-Hermitian parts be $A_h=(A+A^*)/2$ and $A_s=(A-A^*)/(2j)$, respectively. Then $\mathcal{R}_{\geq \theta}(A) = \max\{\sqrt{g^\star},\sqrt{h^\star}\}$, where $g^\star$ and $h^\star$ are respectively the optimal values obtained by solving the following SDP problems:
			\begin{align*}
				&\min_{g,\tau}~ g,~~\text{\rm s.t.}~~ A^*A-g I +\tau(-A_h\tan\theta+A_s) \leq 0,~\tau\geq 0,~g\geq 0;~\text{and}\\
				&\min_{h,\tau}~ h,~~\text{\rm s.t.}~~ A^*A-h I +\tau(-A_h\tan\theta-A_s) \leq  0,~\tau\geq 0,~h\geq 0.
			\end{align*}
		\end{proposition}
		\begin{proof}
			The proof follows similarly as that for Proposition~\ref{prop:compute_constrained_phase}. 
	\end{proof}}
	Based on the constrained gains, a parallel result to Theorem~\ref{thm:DW_matrix} can be obtained as follows.
	\begin{theorem}\label{thm:DW_cutVertically}
		Let $A,B\in\mathbb{C}^{n\times n}$ be sectorial. Then $I+AB$ is invertible if there exists a $\theta\in[0,\pi)$ such that
		$$
		\gamma_\theta(A)\bar{\sigma}(B)<1~\text{and}~\Psi(B)\subset(-\pi+\theta,\pi-\theta).
		$$
	\end{theorem}	
	In particular when $\theta=0$, the above result reduces to the matrix small gain theorem.
	\begin{proof}
		The proof follows similarly as that for Theorem~\ref{thm:DW_matrix} by swapping the roles of gains and phases.
	\end{proof}
	
	\subsection{Feedback Stability by DW-Shells}
	The following result on feedback stability of LTI systems is a direct consequence of the matrix result Theorem~\ref{thm:DW_matrix}. 
	\begin{theorem}\label{thm:DW_pv_LTI}
		Let $P$ be {semi-stable frequency-wise semi-sectorial} with $j\Omega_p$ being the set of poles on the imaginary axis and $C\in\mathcal{RH}_\infty$ be frequency-wise sectorial. Then $P\cls C$ is stable if for all $\omega\in[0,\infty]\setminus\Omega_p$, there exists an $r:=r(\omega) \geq 0$ such that $P(j\omega)$ is $r$-sectorial and
		\begin{align}\label{eq:thm_vp_LTI}
			&\bar{\psi}_r(P(j\omega))+\bar{\phi}(C(j\omega))<\pi,~\underline{\psi}_r(P(j\omega))+\underline{\phi}(C(j\omega))>-\pi,~\bar{\sigma}(C(j\omega))r(\omega)<1.
		\end{align}
	\end{theorem}
	\begin{proof}
		Using Proposition~\ref{prop:monotone_dw_phase}, we have
		\begin{align*}
			&\bar{\psi}_r(\tau P(j\omega))+\bar{\phi}( C(j\omega))<\pi~~\text{and}~~\underline{\psi}_r(\tau P(j\omega))+\underline{\phi}(C(j\omega))>-\pi,
		\end{align*}
		where $\tau\in(0,1]$. It then follows from Theorem~\ref{thm:DW_matrix} that $I+\tau P(j\omega)C(j\omega)$ is invertible all $\omega\in[0,\infty]\setminus\Omega_p$. By the conjugate symmetric property of real rational transfer matrices, namely $G(j\omega) = \overline{G(-j\omega)}$, we obtain that $I+\tau P(j\omega)C(j\omega)$ is invertible for all $\omega\in[-\infty,\infty]\setminus\Omega_p$. A similar argument on $\omega\in \Omega_p$ as that in the proof of Theorem~\ref{thm:mixture} yields that $(I+\tau PC)^{-1}\in \mathcal{L}_\infty$ as well as $G_\tau:=(\tau P)\cls C\in \mathcal{L}_\infty$ for all $\tau\in[0,1]$. When $\tau=0$, $G_0$ is stable since $C$ is stable. As $\tau$ continuously increases from 0 to 1, the poles of $G_\tau$ cannot cross the imaginary axis as $G_\tau\in\mathcal{L}_\infty$, whereby $G_1=P\cls C$ is stable.  
	\end{proof}
	
	The above result has the following parallel, which is mainly attributed to Theorem~\ref{thm:DW_cutVertically} based on the notion of constrained gains. 
	
	\begin{theorem}\label{thm:DW_ph_LTI}
		Let $P$, $C\in\mathcal{RH}_\infty$ and $C$ be frequency-wise sectorial. Then $P\cls C$ is stable if for all $\omega\in[0,\infty]$, there exists an $\theta:=\theta(\omega) \in[0,\pi)$ such that 
		\begin{align}\label{eq:thm_DWph_LTI}
			\gamma_\theta(P(j\omega))\bar{\sigma}(C(j\omega))<1~\text{and}~\Psi(C(j\omega))\subset(-\pi+\theta(\omega),\pi-\theta(\omega)).
		\end{align}
	\end{theorem}
	\begin{proof}
		Using Proposition~\ref{prop:monotone_dw_phase}, we have
		\begin{align*}
			\gamma_\theta(P(j\omega))\bar{\sigma}(\tau C(j\omega))<1~\text{and}~\Psi(\tau C(j\omega))\subset(-\pi+\theta(\omega),\pi-\theta(\omega)),
		\end{align*}
		where $\tau\in(0,1]$. It then follows from Theorem~\ref{thm:DW_cutVertically} that $I+\tau P(j\omega)C(j\omega)$ is invertible all $\omega\in[0,\infty]$. By the conjugate symmetric property of real rational transfer matrices, namely $G(j\omega) = \overline{G(-j\omega)}$, we obtain that $I+\tau P(j\omega)C(j\omega)$ is invertible all $\omega\in[-\infty,\infty]$, whereby $(I+\tau PC)^{-1}\in \mathcal{L}_\infty$ for all $\tau\in[0,1]$ and so does $G_\tau:=(\tau P)\cls C$. The rest follows by the same argument in the proof for Theorem~\ref{thm:DW_pv_LTI}. 
	\end{proof}
	\section{Bounded \& Sectored Real Lemma}\label{sec:bounded_sectored_lemma}
	It is well known that the $\mathcal{H}_\infty$ norm of LTI systems can be easily computed by the bounded real lemma \cite[Chapter~12]{Zhou1998Essential}. Recently, the sectored real lemma, as a counterpart to the bounded real lemma, has been developed in \cite{wei2021phaseLTI} for the computation of system $\Phi_\infty$ phase. Actually, by a suitable application of the generalized KYP lemma \cite{Hara2005GeneralizedKYP}, we can obtain the following result that captures a mixture of gain and phase properties with LMIs. To that end, we start with some preliminary knowledge on the generalized KYP lemma.
	
	Define a curve in the complex plane via
	\begin{align}\label{eq:curve}\Lambda(\Phi,\Psi) = \left\{\lambda\in\mathbb{C}~\bigg|~\begin{bmatrix}
			\lambda \\ 1
		\end{bmatrix}^*\Phi\begin{bmatrix}
			\lambda \\ 1
		\end{bmatrix} = 0,~\begin{bmatrix}
			\lambda \\ 1
		\end{bmatrix}^*\Psi\begin{bmatrix}
			\lambda \\ 1
		\end{bmatrix}\geq 0,
		\right\}\end{align}
	with parameters $\Phi,\Psi$ being Hermitian matrices. 
	According to \cite[Section~IV]{Hara2005GeneralizedKYP}, we obtain the following specialization of curves. 
	\begin{lemma}\label{lem:curves}
		Let $$\omega_c\in(0,\infty),~\Phi=\begin{bmatrix}0 & 1 \\ 1 & 0\end{bmatrix},~\Psi_1=\begin{bmatrix}-2 & j\omega_c \\ -j\omega_c & 0\end{bmatrix},~\text{and}~\Psi_2=\begin{bmatrix}0 & j \\ -j & -2\omega_c\end{bmatrix}.$$
		It holds that
		$$\Lambda(\Phi,\Psi_1) = j[0,\omega_c]~~\text{and}~~\Lambda(\Phi,\Psi_2) = j[\omega_c,\infty).$$
	\end{lemma}
	The following generalized KYP lemma is tailored from \cite[Theorem~1]{Hara2005GeneralizedKYP}.
	\begin{lemma}\label{lem:generalizedKYP}
		Let $A\in\mathbb{C}^{n\times n}$, $B\in\mathbb{C}^{n\times m}$, $M=M^*\in\mathbb{C}^{(n+m)\times(n+m)}$, and $\Lambda(\Phi,\Psi)$ be introduced in \eqref{eq:curve}. Suppose $\left[\;\,\begin{matrix} A&\vline&B\end{matrix}\;\,\right]$ is controllable. Let $\Omega$ be the set of eigenvalues of $A$ in $\Lambda(\Phi,\Psi)$. Then 
		$$\begin{bmatrix}
			(\lambda I-A)^{-1}B \\ I
		\end{bmatrix}^*M\begin{bmatrix}
			(\lambda I-A)^{-1}B \\ I
		\end{bmatrix}\leq 0$$
		for all $\lambda\in\Lambda(\Phi,\Psi)\setminus\Omega$ if and only if there exist Hermitian matrices $P$ and $Q$ such that $Q\geq 0$ and
		\begin{align*}
			&\begin{bmatrix} A & B \\I & 0 \end{bmatrix}^*(\Phi\otimes P + \Psi\otimes Q) \begin{bmatrix} A & B \\I & 0 \end{bmatrix}
			+M\leq 0.
		\end{align*}
	\end{lemma}
	\begin{theorem}\label{thm:boundedSectoredRealLemma}
		Let $\omega_c>0$ and $G$ be {semi-stable frequency-wise semi-sectorial} over $(-\omega_c,\omega_c)$ with $j\Omega_p$ being the set of poles on the imaginary axis satisfying $\max_{\omega\in\Omega_p} |\omega|<\omega_c$. A minimal realization of $G$ is given by $\left[\;\,\begin{matrix} A&\vline&B\\\hline\\[-4.7mm]C&\vline&D\end{matrix}\;\,\right]$, $\gamma\in(0,\infty)$, $\alpha,\beta\in(-\pi,\pi)$ with $\beta-\alpha\in(0,\pi]$,
		$$\Pi_1:=\begin{bmatrix}0  & e^{j(\alpha-\frac{\pi}{2})}\\ e^{-j(\alpha-\frac{\pi}{2})} & 0  \end{bmatrix},~\Pi_2:=\begin{bmatrix}0  & e^{j(\frac{\pi}{2}+\beta)}\\ e^{-j(\frac{\pi}{2}+\beta)} & 0  \end{bmatrix},
		~\text{and}~\Pi_3:= \begin{bmatrix}I & 0\\ 0 & -\gamma^2 I  \end{bmatrix}.$$
		Then $\Psi(G(j\omega))\subset[\alpha,\beta]$ for $\omega \in [0,\omega_c]\setminus\Omega_p$, and $\bar{\sigma}(G(j\omega))<\gamma$ for $\omega\in [\omega_c,\infty)$  if and only if there exist Hermitian matrices  $P_i,Q_i$, such that $Q_i \geq 0$, $i=1,2,3$, and 
		\begin{equation}\label{eq:thm4}
			\begin{aligned}
				&\begin{bmatrix} A & B \\I & 0 \end{bmatrix}^*\begin{bmatrix}-2Q_i & P_i+j\omega_c Q_i\\P_i-j\omega_c Q_i & 0 \end{bmatrix} \begin{bmatrix} A & B \\I & 0 \end{bmatrix}
				+\begin{bmatrix}C & D\\ 0 & I  \end{bmatrix}^*\Pi_i\begin{bmatrix}C & D\\ 0 & I  \end{bmatrix} \leq 0,~i=1,2,\\
				&\begin{bmatrix} A & B \\I & 0 \end{bmatrix}^*\begin{bmatrix}0 & P_3+j Q_3\\ P_3-j Q_3 & -2\omega_c Q_3 \end{bmatrix} \begin{bmatrix} A & B \\I & 0 \end{bmatrix}
				+\begin{bmatrix}C & D\\ 0 & I  \end{bmatrix}^*\Pi_3\begin{bmatrix}C & D\\ 0 & I  \end{bmatrix} \leq 0.
			\end{aligned}
		\end{equation}
		
	\end{theorem}
	\begin{proof}
		Apply Lemma~\ref{lem:generalizedKYP} three times with 
		$$\Phi_1=\begin{bmatrix}0 & 1 \\ 1 & 0\end{bmatrix},~\Psi_1=\begin{bmatrix}-2 & j\omega_c \\ -j\omega_c & 0\end{bmatrix},$$
		$$\Phi_2=\begin{bmatrix}0 & 1 \\ 1 & 0\end{bmatrix},~\Psi_2=\begin{bmatrix}-2 & j\omega_c \\ -j\omega_c & 0\end{bmatrix}$$	
		and
		$$\Phi_3=\begin{bmatrix}0 & 1 \\ 1 & 0\end{bmatrix},~\Psi_3=\begin{bmatrix}0 & j \\ -j & -2\omega_c\end{bmatrix},$$
		respectively. Together with Lemma~\ref{lem:curves}, we then obtain that the inequalities in \eqref{eq:thm4} are equivalent to that for $\lambda\in\Lambda(\Phi,\Psi)\setminus\Omega_p$ 
		$$\begin{bmatrix}
			(\lambda I-A)^{-1}B \\ I
		\end{bmatrix}^*\begin{bmatrix}C & D\\ 0 & I  \end{bmatrix}^*\Pi_i\begin{bmatrix}C & D\\ 0 & I  \end{bmatrix}\begin{bmatrix}
			(\lambda I-A)^{-1}B \\ I
		\end{bmatrix}\leq 0,~i=1,2,3,$$
		which are further equivalent, via $G(s) = C(sI-A)^{-1}B+D$, to that $\Psi(G(j\omega))\subset[\alpha,\beta]$ for $\omega \in [0,\omega_c]$ and that $\bar{\sigma}(G(j\omega))\leq\gamma$ for $\omega\in [\omega_c,\infty)$. 
	\end{proof}

	The theorem reduces to the bounded real lemma \cite[Chapter~12]{Zhou1998Essential} if $\omega_c\to 0$, or to the sectored real lemma \cite{wei2021phaseLTI} if $\omega_c\to \infty$.  
	
	Theorem~\ref{thm:boundedSectoredRealLemma} gives a state-space characterization of mixed gain/phase properties of LTI systems in terms of a triple of LMIs. By appropriately choosing parameters required by the generalized KYP lemma, we can extend the current result to incorporate more complicated gain/phase restrictions.

	\bibliography{mixgp}
\end{document}